\documentclass[journal,comsoc]{IEEEtran}[10 pt]
\usepackage{amssymb}
\usepackage{amsmath}
\usepackage{textcomp}
\usepackage{cite}
\usepackage{graphicx}	
\usepackage{subfigure}
\usepackage{epsfig}
\usepackage{caption}
\usepackage{grffile}

\usepackage{setspace}
\usepackage{tikz}
\usetikzlibrary{positioning,chains,fit,shapes,calc}
\usetikzlibrary{arrows}
\usepackage{tkz-graph}
\usepackage{color}
\usepackage{multirow}
\usepackage{amsthm}
\usepackage{kbordermatrix}
\newtheorem{mylemma}{Lemma}
\newtheorem{mytheo}{Theorem}

\newtheorem{mydef}{Definition}
\newtheorem{myobser}{Observation}
\newtheorem{mycor}{Corollary}
\usepackage[linesnumbered,ruled,vlined]{algorithm2e}

\title{Properties of a Projected Network of a Bipartite Network}

\author{Suman Banerjee,
	Mamata Jenamani
	and
	Dilip Kumar Pratihar
	\thanks{Suman Banerjee is with the Department
		of Industrial and Syatems Engineering at Indian Institute of Technology, Kharagpur,
		West Bengal, India. e-mail: banerjeesuman1991@gmail.com.}
	\thanks{Mamata Jenamani is with the Department
		of Industrial and Syatems Engineering at Indian Institute of Technology, Kharagpur,
		West Bengal, India. e-mail:mj@iem.iitkgp.ernet.in.}
	\thanks{Dilip Kumar Pratihar is with the Department
		of Mechanical Engineering at Indian Institute of Technology, Kharagpur,
		West Bengal, India. e-mail: dkpra@mech.iitkgp.ernet.in.}}

\begin{document}

\maketitle

\begin{abstract}
Bipartite Graph is often a realistic model of complex networks where two different sets of entities are involved and relationship exist only two entities belonging to two different sets. Examples include the \textit{user-item} relationship of a recommender system, \textit{actor-movie} relationship of an online movie database systems. One way to compress a bipartite graph is to take unweighted or weighted one mode projection of one side vertices. Properties of this projected network are extremely important in many practical situations (say the selection process of influencing nodes for viral marketing). In this paper, we have studied the topological properties for  projected network and some theoretical results are proved including the presence of \textit{cliques}, \textit{connectedness} for unweighted projection and \textit{maximum edge weight} for weighted projected network. 
\end{abstract}


\section{Introduction} \label{Sec:Introduction}
Bipartite graph \footnote{In this paper, the terms Graph and Network used interchangeably.} \cite{bondy1976graph} is a kind of graph, whose vertex set can be partitioned into two disjoint subsets and two end points of each edge will be from two different vertex sets. Two disjoint sets of entities and their relationships can be easily modeled as a bipartite network. That is why many networks from social to biological has inherent bipartite structure, such as user\mbox{-}item relationship of an Online Social Rating  Networks (OSRN)\cite{matsuo2009community}, customer\mbox{-}product relationship of a recommender systems \cite{ricci2011introduction} \cite{huang2007analyzing}, author\mbox{-}book relationship of a library database systems, citation network between researchers and research papers \cite{gustafsson2014describing}, affiliation network \cite {lattanzi2009affiliation} of professors and academic institutions etc. 
Similarly, in biological contexts there are many situations, where bipartite network appears naturally like the relationship between metabolites and enzymes i.e. which enzyme would act on which metabolites \cite{ravasz2002hierarchical}, the relationship between genes and proteins \cite{vazquez2002modeling}, gene-diseases network \cite{ozgur2008identifying}, i.e., which diseases affects which genes etc.

Though for analyzing a one mode network there exist a lot of techniques, measures and algorithms but for bipartite network it is limited. Projection \cite{zweig2011systematic} (formally defined in section 4) is often used to compress a bipartite network to obtain unipartite network, which can be analyzed further. Practically, one mode projection is widely used in various contexts. Zhoe et al. \cite{zhou2007bipartite} used the one mode projection of a customer\mbox{-}product network for item recommendation. Shang et al. \cite{shang2008personal} used one mode projection of a weighted bipartite graph for item recommendation which produces better recommendation accuracy compared to some other similarity based methods. Besides, in literature the user network which we obtained by one mode projection of user-item bipartite network of an online social rating network, can be used for information diffusion \cite{guille2013information} \cite{bakshy2012role} to recommend items \cite{an2016diffusion}. Different kinds of analysis are required for the user network with different Topological properties. In this paper we have made a theoretical study of the projected network obtained from a bipartite graph and other kind of works like information diffusion etc. will be taken up subsequently. But, the problem lies with one mode projection is that, we loose much information about the parent bipartite graph. One of the ways to retain more information compared to one mode peojection is the weighted one mode projection and the network obtained after weighted one mode projection is called \textit{weighted projected Network}. In this paper, we have given an algorithm for computing the projected network from a given bipartite graph. And also we have studied different topological properties like connectedness, presence of cliques etc. about the projected network.

Rest of this paper is organized as follows: Section \ref{Sec:Preliminary} of this paper talks about preliminary definition and terminologies of graph theory and Bipartite graphs with examples. In Section \ref{Sec:OB}, we have made some crucial observations about bipartite graphs and there we have introduced the notion of sparse and dense bipartite graphs. In Section \ref{Sec:OMP}, the notion of weighted and unweighted projection of a bipartite network has been defined formally with examples. In Section \ref{Sec:WOMP}, we have studied the properties of projected network which includes the connectedness, presence of cliques and upper bound on the edge weight. Finally, we draw conclusions of our work and give future directions in Section \ref{Sec:CFD}. 
\section{Preliminary Concepts} \label{Sec:Preliminary}
\subsection{Basic Terminologies}
 A \textit{graph} is generally represented by a two tuple $G(V,E)$, where $V(G)$ is known as \textit{vertex} set and edge set $E(G)$, where $E(G) \subset V(G) \times V(G)$. \textit{Cardinalities} of $V(G)$ and $E(G)$ are  respectively known as order and edge count of $G$. Two vertices $v_i$ and $v_j$ are said to be \textit{adjacent} if they are connected by an edge, i.e., $(v_iv_j) \in E(G)$. An edge is \textit{incident} on a vertex, if it is one of the end vertex of that edge. Two edges are adjacent to each other, if they have one end vertex in common. For any arbitrary vertex $v_i$ of $G$, other vertices directly connected by an edge with $v_i$ is known as the \textit{neighborhood} of $v_i$ in $G$ and denoted by $N(v_i)$. Cardinality of neighborhood of a vertex is called the degree of the vertex and denoted by $deg(v_i)$ so $deg(v_i)=\vert N(v_i) \vert$. A vertex with degree 0 is called isolated vertex, and degree 1 is called pendent vertex. Two edges have the same end vertices is known as parallel edges. For any edge, if both the end vertices are identical, then the edge is called as \textit{self \mbox{-} loop}. A path in a graph is a sequence of vertices connected by edges where no vertex is repeated. A graph is connected if between every pair of vertices there exist at least one path. A graph if contains either parallel edges or self \mbox{-} loops or both is known as multigraph. A graph is called as simple graph if it neither contains parallel edges nor self \mbox{-} loops. A graph is called a weighted graph if its edges are labeled with a real number. Weighted graphs are generally symbolized as $G(V, E, W)$ where $V(G)$ and $E(G)$ are vertex set and edge set of the graph respectively as usual and $W(G)$ is a function which assigns a real number to every edge of $G$ mathematically $W(G): E(G)\rightarrow R$. For other terminologies readers may refer to \cite{diestel2000graph}. In this work, we have considered only connected bipartite graphs. 
\subsection{Bipartite Graph} \label{Sec:BG}
A graph will be said to a bipartite graph if its vertex set has two disjoint subsets and each edge connects one vertex from each group or in other words both the end vertex of any edge will not be from the same subset. Mathematically the definition can be stated as follows {\cite{abiad2013algebraic}:
\begin{mydef}
$G(U,S,E)$ will be called a bipartite graph if $V(G)= U(G) \cup S(G)$ and $U(G) \cap S(G) = \phi$ and for each edge $(uv) \in E(G)$ either $u \in U(G)$, $v \in S(G)$ or $v \in U(G)$, $u \in S(G)$. $G$ will be a complete bipartite graph if $\forall u \in U(G)$ and $\forall v \in S(G)$, $(uv) \in E(G)$. 
\end{mydef} 
For a bipartite graph it is interesting to notice that for any vertex $u_i \in U$, $N(u_i) \subset S$ and $\forall s_j \in S$ $N(s_j) \subset U$. Normally a bipartite graph is represented by a $\vert U(G) \vert \times \vert S(G) \vert$ matrix known as \textit{Bi-adjacency} matrix ($B$) and the content of the matrix is as follows: 
\[
    (B)_{ij}= 
\begin{cases}
    1,& \text{if } (u_is_j)\in E(G)\\
    0,              & \text{otherwise}
\end{cases}
\]
Below a bi-adjacency matrix of order $6 \times 4$ is given and its corresponding bipartite graph is shown in Fig. 1:\\
\renewcommand{\kbldelim}{(}
\renewcommand{\kbrdelim}{)}
\[
  M = \kbordermatrix{
    & s_1 & s_2 & s_3 & s_4  \\
    u_1 & 1 & 0 & 0 & 0 \\
    u_2 & 1& 1 & 0 & 0 \\
    u_3 & 1& 1 & 1 & 0 \\
    u_4 & 0 & 1 & 0 & 1\\
    u_5 & 0 & 0 & 1 & 0\\
    u_6 & 0 & 0 & 0 & 1
  }
\]
\begin{figure}[h]
\begin{center}
\definecolor{myblue}{RGB}{80,80,160}
\definecolor{mygreen}{RGB}{80,160,80}
\begin{tikzpicture}[thick,
  every node/.style={draw,circle},
  every loop/.style={},
  fsnode/.style={fill=myblue},
  ssnode/.style={fill=mygreen},
  every fit/.style={ellipse,draw,inner sep=-2pt,text width=2cm},
  shorten >= 3pt,shorten <= 3pt
]
\centering
\begin{scope}[start chain=going below,node distance=7mm]
\foreach \i in {1,2,...,6}
  \node[fsnode,on chain] (f\i) [label=left: \ $u_\i$] {};
\end{scope}

\begin{scope}[xshift=4cm,yshift=-0.5cm,start chain=going below,node distance=5mm]
\foreach \i in {1,2,...,4}
  \node[ssnode,on chain] (s\i) [label=right: \ $s_\i$] {};

\end{scope}
\node [myblue,fit=(f1) (f6),label=below:$U$] {};
\node [mygreen,fit=(s1) (s4),label=below:$S$] {};
\draw (f1) - - (s1);
\draw (f2) - - (s1);
\draw (f3) - - (s1);
\draw (f2) - - (s2);
\draw (f3) - - (s2);
\draw (f4) - - (s2);
\draw (f3) - - (s3);
\draw (f4) - - (s4);
\draw (f5) - - (s3);
\draw (f6) - - (s4);
\end{tikzpicture}
\caption{A Bipartite Graph} 
\end{center}
\end{figure}
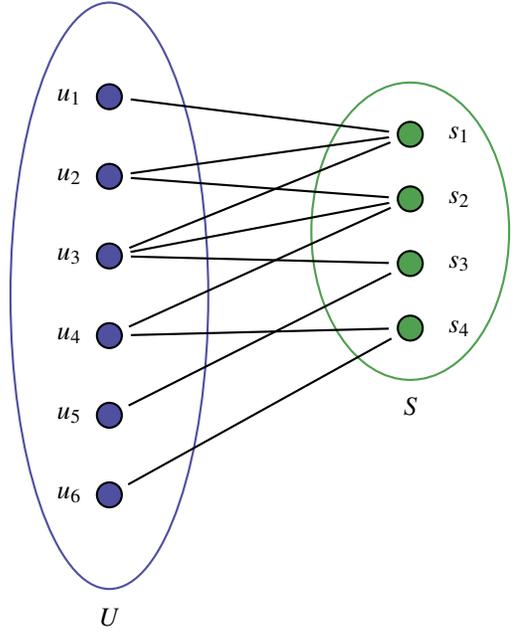
\section{Observations on Bipartite Graphs} \label{Sec:OB}
We make the following observations in the context of Bipartite Graph. These observations form the basis for deriving the properties of the corresponding unipertite network. 
\begin{myobser}
For a bipartite graph, sum of degress of two different sides of vertices will be equal. Mathematically,  for a bipartite graph $G(U, S, E)$, where $\vert U(G) \vert=n_1$, $\vert S(G) \vert =n_2$ and $\vert E(G) \vert=m$, then the following relation always holds:\\
\begin{equation}
\sum_{i=1}^{n_1}deg(u_i)=\sum_{j=1}^{n_2}deg(s_j)
\end{equation} 
\end{myobser}
\begin{myobser}
For a connected bipartite Graph $G(U, S, E)$ with $\vert U(G) \vert=n_1$, $\vert S(G) \vert =n_2$ then $\forall u_i \in U$, $1 \leq deg(u_i) \leq n_2$ and $\forall s_j \in S$, $1 \leq deg(s_j) \leq n_1$.
\end{myobser}
If $\vert U(G) \vert =n_1$ and $\vert S(G) \vert=n_2$ then maximum possible edges in the bipartite graph will be $n_1n_2$, i.e., $\vert E(G) \vert _{max}=n_1n_2$. A simple graph $G(V, E)$ is said to be a sparce if  its number of edges present in the graph is linear with the number of vertices i.e. $\vert E(G) \vert \simeq \mathcal{O}(\vert V(G) \vert)$ and dense if its no. of edges is quadratic with no. of vertices i.e.  $\vert E(G) \vert \simeq \mathcal{O}(\vert V(G) \vert^{2})$ \cite{cormen2009introduction}.\\
This defination can be transformed in the context of bipartite graph as follows:
\begin{mydef}
Let $G(U,S,E)$ be a bipartite graph with $\vert U(G) \vert=n_1$, $\vert S(G) \vert =n_2$ and $\vert E(G) \vert=m$. Now, $G$ will be a sparse bipartite graph if $m=\mathcal{O}(n_1 + n_2)$ and dense bipartite graph if $m=\mathcal{O}(n_1n_2)$
\end{mydef}  
\begin{mytheo}[Euler]
\cite{bollobas2013modern} For a simple graph sum of degrees of all the vertices will be twice the number of edges. So, if $\mathcal{G}(V, E)$ be a simple graph with $\vert V(\mathcal{G}) \vert = n$ and $\vert E(\mathcal{G}) \vert = m$ then,
\begin{equation}
\sum_{i=1}^{n} deg(v_i)=2m
\end{equation}
\end{mytheo}
In the context of bipartite graphs Equation no. (2) can be written as follows:\\
\begin{equation}
\sum_{i=1}^{n_1} deg(u_i) + \sum_{j=1}^{n_2} deg(s_j)=2m
\end{equation}
Now, using Equation no. (1) and (3) we can write:\\
\begin{equation}
\sum_{i=1}^{n_1} deg(u_i)=m
\end{equation}
and 
\begin{equation}
\sum_{j=1}^{n_2} deg(s_j)=m
\end{equation}
If $G$ the is a sparse bipartite graph then by the definition of sparse bipartite graph from Equation no. (5)
\begin{equation}
\sum_{i=1}^{n_1} deg(u_i)=\mathcal{O}(n_1 + n_2)
\end{equation}
By the definition of Big-oh notation Equation no. (6) can be written as follows:\\
\begin{equation}
\sum_{i=1}^{n_1} deg(u_i) \leq\ c(n_1 + n_2)
\end{equation}
and similarly
\begin{equation}
\sum_{j=1}^{n_2} deg(s_j) \leq\ c(n_1 + n_2)
\end{equation}
where $c$ is a positive constant.
In many practical situations $n_1 \gg n_2$. Like if we consider the scenario of an online social rating network there number of the users is much more than the number of items. In that context, inequalities (7) and (8) can be reduced to:\\ 
\begin{equation}
\sum_{i=1}^{n_1} deg(u_i) \leq\ c n_1 
\end{equation}
and similarly,
\begin{equation}
\sum_{j=1}^{n_2} deg(s_j) \leq\ c n_1 
\end{equation}
\section{Weighted and Unweighted One Mode Projection} \label{Sec:OMP}
As most of the network analysis algorithms and measures are mainly designed for general kind of graphs, one of the usual technique is to project the one side of the vertices of a bipartite graph based on the connectivity with the other side of vertices. Formally one mode projection of a bipartite graph can be defined as follows  
\begin{mydef}
Let $G(U,S,E)$ be a bipartite graph with $\vert U(G) \vert=n_1$, $\vert S(G) \vert =n_2$ and $\vert E(G) \vert=m$. Now projection of the bipartite graph $G$ for the vertex set $U$ with respect to the vertex set $S$ is to construct a unipertite or one mode network $G^{'}(U, E^{'})$ where $V(G)=U$ and $(u_iu_j) \in E(G^{'})$ if $N(u_i) \cap N(u_j) \neq \phi$.
\end{mydef}
From a bipartite graph always two projected network will be obtained. One for the projection of the Vertex set $U$ with respect to the vertex set $S$ and the other one for projecting the vertex set $S$ with respect to the vertex set $U$. In rest of our paper unless otherwise stated all the theoretical results that have been proved for the projected network obtained by the projection of the vertex set $U$ with respect to the vertex set $S$.\\ 
Following algorithm will take the Bi-adjacency matrix ($B$) of bipartite graph ($G$) and produces the adjacency matrix ($A$) of the projected network ($G^{'}$):\\

\begin{algorithm}
 \KwData{Bi-adjacency Matrix ($B$) of the Bipartite Network.}
 \KwResult{ Adjacency Matrix ($A$) of the Projected Network}
 $n_1 \leftarrow B.no\_ of\_rows()$\;
 $n_2 \leftarrow B.no\_ of\_columns()$\;
 \For{$i=1 $ to $n_1$}{
 \For{$j=i+1 $ to $n_1$}{
 \For{$k=1 $ to $n_2$}{
  \eIf{$B[i][k]==1 \&\& B[j][k]==1$}{
   $A[i][j]=1$\;
   break\;
   }{
   $A[i][j]=0$\;
  }
 }
 }
 }
 \caption{Algorithm for Computing Projected Network}
\end{algorithm}
Algorithm 1 mainly performing the exhaustive search for all possible vertex pairs of one side whether they have at least one common neighborhood vertex in the other side or not. So, computational time required by Algorithm 1 is as follows:\\
\begin{center}
$f(n_1n_2)= {n_1 \choose 2} \mathcal{O}(n_2)$\\

$\Rightarrow f(n_1n_2)=  \frac{n_1(n_1-2)}{2} \mathcal{O}(n_2)$\\

$\Rightarrow f(n_1n_2)= \mathcal{O}(n_1^{2}n_2)$
\end{center}
and except some loop variables one matrix of dimension $n_1 \times n_1$ has been created in Algorithm 1.
\begin{mytheo}
Algorithm has running time $\mathcal{O}(n_1^{2}n_2)$ and space requirement $\mathcal{O}(n_1^{2})$.
\end{mytheo}
So, by the above\mbox{-}maintained definition if $A$ be the adjacency matrix of the graph $G^{'}$ then the content of $A$ can be written as follows:\\
\[
    A_{ij}= 
\begin{cases}
    1,& \text{if } N(u_i) \cap N(u_j) \neq \phi \\
    0,              & \text{otherwise}
\end{cases}
\]
For the bipartite graph shown in section \ref{Sec:BG} its projected network for the vertex set $U$ with respect to the vertex set $S$ and for the vertex set $S$ with respect to the vertex set $U$ is shown in Fig. 2 and Fig. 3 respectively.

\begin{figure}
\begin{center}
\includegraphics[scale=0.7]{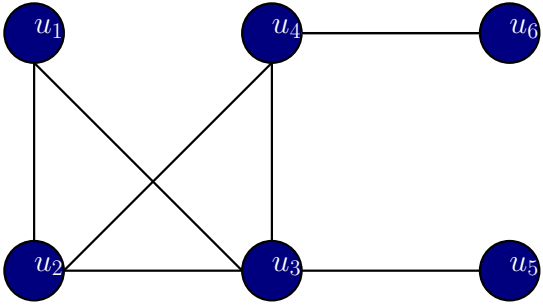}
\caption{Projected Network of U with respect to S}
\end{center}
\end{figure}

\begin{figure}
\begin{center}
\includegraphics[scale=0.7]{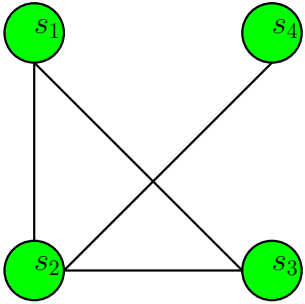}
\caption{Projected Network of S with respect to U}
\end{center}
\end{figure}

Weighted one mode projection is a technique used to obtain a weighted one mode network from a bipartite network, where edge weight represents the number of common neighbor the vertices have and the network obtained using weighted one mode projection is called weighted projected network.
Formally it can be defined as follows:
\begin{mydef}
For a bipartite graph $G(U, S, E)$ with $\vert U \vert=n_1$, $\vert S \vert=n_2$, $\vert E \vert=m$, weighted one mode projection of $G$ for the vertex set $U$ with respect to the vertex set $S$ will be a weighted unipartite network $\mathcal{G}(U, \mathcal{E}, W)$, where $V(\mathcal{G})=U$, two vertex $u_i,u_j \in U$, will be a edge in $\mathcal{G}$ i.e. $(u_iu_j) \in E(\mathcal{G})$ if $N(u_i) \cap N(u_j) \neq \phi$ in $G$ and $W$ is a weight function which assigns the number of common neighbor two vertices have, i.e. $W: (u_iu_j) \rightarrow \vert N(u_i) \cap N(u_j) \vert$. 
\end{mydef} 
From the above definition the content of the weight matrix ($W$) of the weighted projected network ($\mathcal{G}$) can be written as follows:
\[
    W_{ij}= 
\begin{cases}
    \vert N(u_i) \cap N(u_j) \vert,& \text{if } N(u_i) \cap N(u_j) \neq \phi \\
    0,              & \text{otherwise}
\end{cases}
\]
With minor modification Algorithm 1 can be used for computing weighted one mode projection as well.\\

%
%
%

%
%
\section{Properties of Weighted and Unweighted Projected Network} \label{Sec:WOMP}
In this section, we have studied the properties of unweighted projected network for connectedness, presence of cliques and proved the upper bound on the edge weight of the weighted projected network.
\begin{mylemma}
If $G^{'}(U, E^{'})$ be the projected network of a bipartite graph $G(U, S, E)$ for the vertex set $U$ with respect to the vertex set $S$ then each vertex $s_j \in S$ induces a clique of size $\vert N(s_j) \vert$ if $\vert N(s_j) \vert \geq 2$.  
\end{mylemma}
\begin{proof}
This lemma can be proved from some simple observation from the Definition 3 given in Section \ref{Sec:OMP}. As $\forall u_i \in N(s_j)$ has atleast one common neighbor in $S$ and that is $s_j$. So, $\forall u_i,u_j \in N(s_j)$ $ N(u_i) \cap N(u_j) \neq \phi$. And, hence $N(s_j)$ will be a clique in $G^{'}$. This can be verified from the bipartite graph shown in Fig. 1. also. In Fig. 1. for $s_2 \in S$ its neughborhood $N(s_2)=\left \{u_2,u_3,u_4 \right \}$ and that's why $\left \{u_2,u_3,u_4 \right \}$ is a clique in the graph shown in Fig. 2. Similarly, for $u_3 \in U$, $N(u_3)=\left \{s_1,s_2,s_3 \right \}$ and so $\left \{s_1,s_2,s_3 \right \}$ is a clique in the graph shown in Fig. 3.
\end{proof}
\begin{mylemma}
If $G^{'}(U, E^{'})$ be the projected network of a bipartite graph $G(U, S, E)$ for the vertex set $U$ with respect to the vertex set $S$ and for some $(u_is_j)\in E(G)$ $(u_is_k)\notin E(G)$ for $k \in [n_2]\backslash j$ and also $(u_ps_j)\notin E(G)$ for $p \in [n_1]\backslash i$ then $G^{'}$ will be disconnected. \footnote{$[n]$ means the set of natural numbers from $1$ to $n$ i.e.$\{ 1, 2, 3, \dots\ , n\}$}
\end{mylemma}
\begin{proof}
Let us consider  that $G(U, S, E)$ be a bipartite graph with vertex set $U=\{ u_1, u_2, u_3, \dots\ , u_{n_1}\}$ and $S=\{ s_1, s_2, s_3, \dots\ , s_{n_2}\}$ and there exist two vertices, one in $U$ ( i.e. $u_i \in U$) and another in $S$ (i.e. $s_j \in S$) such that $(u_is_j) \in E(G)$ and $(u_is_k)\notin E(G)$ for $k \in [n_2]\backslash j$ and also $(u_ps_j)\notin E(G)$ for $p \in [n_1]\backslash i$. It essentially implies that $u_i$ and $s_j$ are pendant vertices in $U$ and $S$ respectively. As $u_i$ is a pendant vertex and only connected with $s_j$ in $S$ so $N(u_i)=\{s_j\}$. Also, it is not difficult to see that as $s_j$ is a pendent vertex in $S$ so $\nexists u \in U$ such that $N(u) \cap N(u_i) \neq \phi$. It essentially implies that $u_i$ will be disconnected from the other part of the network means $G^{'}$ will be disconnected.
\end{proof}
\begin{mylemma}
If $\mathcal{G}(U, \mathcal{E}, W)$ be the weighted projected network of a bipartite network $G(U, S, E)$ and $W_{ij}$ be the edge weight of the edge $(u_iu_j)$ then $1 \leq W_{ij} \leq n_2$ where $n_2=\vert S \vert$.
\end{mylemma}  
\begin{proof}
Lower bound of the edge weight of $\mathcal{G}$ is trivial. By the Definition 4, two vertices will be connected by an edge in the projected network, if they have at least one common neighbor in the parent bipartite graph. That's why $W_{ij} \geq 1$.\\
Upper bound can be proved from the Observation 1 maintained in Section 3. As for any $u_i \in U$, $1 \leq deg(u_i) \leq n_2$, so for any two vertex $u_i,u_j \in U$ if $u_i$ and $u_j$ is connected with all the vertices of $S$ i.e. $\forall k=1, 2, \dots\, n_2$ $(u_is_k) \in E(G)$ and $(u_js_k) \in E(G)$ then $N(u_i)= N(u_j)=S$ which clearly imply $ N(u_i) \cap N(u_j)=S$ hense $\vert N(u_i) \cap N(u_j) \vert =n_2$ and so $W_{ij} \leq n_2$. So, edge weight will be in between $1$ and $n_2$ i.e. $1 \leq W_{ij} \leq n_2$
\end{proof}
\begin{mytheo}
If $\mathcal{G}(U, \mathcal{E}, W)$ be the weighted projected network of a bipartite network $G(U, S, E)$ and $W_{ij}$ denotes the edge weight of the edge $(u_iu_j)$ then $\sum_{(u_iu_j)\in \mathcal{E}} W_{ij}= \sum_{k=1}^{n_2} {d_k \choose 2}$ if $d_k \geq 2$.
\end{mytheo}
\begin{proof}
This theorem can be proved from the concept described in Lemma 1. As each vertex $s_k \in S$ induces a clique of size $\vert N(s_k) \vert$ in $\mathcal{G}$. So, if $\vert N(s_k) \vert \geq 2$ then $s_k$ contributes to the edge weight in $\mathcal{G}$ by ${d_k \choose 2}$. If $\omega_{s_k}$ be the total edge weight due to the vertex $s_k$ then,\\
\[
    \omega_{s_k}= 
\begin{cases}
    {d_k \choose 2},& \text{if } d_k\geq 2\\
    0,              & \text{otherwise}
\end{cases}
\]
So, if we sum it up for all the vertices in $S$ then that will be equal to the sum of the all edge weights in $\mathcal{G}$. So,  $\sum_{(u_iu_j)\in \mathcal{E}} W_{ij}= \sum_{k=1}^{n_2} {d_k \choose 2}$, if $d_k \geq 2$.
\end{proof}
Now, let us take the situation of a complete bipartite graph. When a complete bipartite graph $G(U, S, E)$ is projecting for the vertex set $U$ with respect to the vertex set $S$ then the number of edges in the projected network will be ${n_1 \choose 2}$ and each edge will have edge weight $n_2$. So, total edge weight in the projected network of a bipartite network will be 
\begin{center}
$\sum_{(u_iu_j) \in \mathcal{E}} W_{ij}={n_1 \choose 2} n_2$\\
$\Rightarrow \sum_{(u_iu_j) \in \mathcal{E}} W_{ij}=  \frac{n_1n_2(n_1-1)}{2}$
\end{center} 
and in any other cases total edge weight will be less than $\frac{n_1n_2(n_1-1)}{2}$ which can be formally stated as follows:
\begin{mycor}
If $\mathcal{G}(U, \mathcal{E}, W)$ be the weighted projected network of a bipartite network $G(U, S, E)$ and $\omega$ denotes the sum of all edge weights of $\mathcal{G}$ then $\omega \leq \frac{n_1n_2(n_1-1)}{2}$ .
\end{mycor}
\section{Conclusion and Future Directions}  \label{Sec:CFD}
In this work, we have studied the properties of a projected network generated from a bipartite network. We have found out some interesting properties about connectedness, presence of cliques in the projected network, maximum possible possible edge weight for the weighted projected network. One algorithm has been presented in section \ref{Sec:OMP} which computes the adjacency matrix of the projected network from the given bi-adjacency matrix of a bipartite graph.\\
Now this work can be extended in different directions. First, the algorithm that has been proposed in Section 4 has quadriatic complexity. So, efficient algorithm can be proposed to compute the adjacency matrix of the projected network. Secondly, many real life network data are available for different kinds of networks. Theoretical results that has been studied in this paper that can be applied to study the properties of those real life networks.   
\section*{Acknowledgment}
The work has been financially supported by the project \textit{E-business Center of Excellence} funded by Ministry of Human Resource and Development (MHRD), Government of India under the scheme of \textit{Center for Training and Research in Frontier Areas of Science and Technology (FAST)}, Grant No. F.No.5-5/2014-TS.VII  .
\bibliographystyle{IEEEtran}
\bibliography{refer}

\end{document}